\documentclass[aps,prl,onecolumn,superscriptaddress,floatfix,reprint]{revtex4}
\usepackage[english]{babel}
\usepackage[utf8]{inputenc}
\usepackage{amssymb}
\usepackage{amsthm}
\usepackage{mathtools}
\usepackage{dsfont}
\usepackage{graphicx}
\usepackage{cancel}

\usepackage[draft,inline,nomargin]{fixme} \fxsetup{theme=color}
\newcommand{\tr}{{\rm tr}\,}
\newcommand{\ket}[1]{\left|{#1}\right\rangle}
\newcommand{\bra}[1]{\left\langle{#1}\right|}
\newcommand{\braket}[2]{\langle{#1}|{#2}\rangle}

\newcommand{\ketbra}[2]{\left|{#1}\rangle\!\langle{#2}\right|}

\newtheorem{theorem}{Theorem}

\newtheorem{proposition}{Proposition}

\begin{document}

\title{Fixed Point Engineering of Quantum Batteries}

\author{Esteban Mart\'inez Vargas}
\affiliation{Calle centella 3, Smza 18 Mza 3, 77505 Cancún México}
\email{estebanmv@protonmail.com}

\begin{abstract}
    We develop, in the context of open quantum systems, a general
    formalism to build completely positive and unital (CPU) maps that
    have a desired fixed point. In contrast to quantum channels, 
    these maps act on observables, not on states. As the CPU maps
    will have a known observable as their fixed point, we know by
    construction that in the asymptotic limit the map will yield a
    constant observable. Therefore, if we implement this in a
    differential equation, we have a linear growth of a quantity 
    with respect to time. Applying these ideas to the charging of quantum
    batteries, we obtain an upper bound on the charging power of 
    batteries by allowing interchanges of multiple states every
    time the CPU map is applied.
\end{abstract}

\maketitle
\section{Introduction}
Electronics have filled our lives with lots of devices, from
smartphones and computers all the way to contemporary electric cars.
Many of these devices function through stored energy in batteries.
Usual batteries work through a chemical reaction, creating an
electric potential difference \cite{VincentScrosatiMdrnBatt1997}.
Nowadays, as we consider relevant quantum technology like quantum computation \cite{NielsenChuangQuantum2011},
there exists the motivation to ask the question of
using quantum devices to store energy. It would also be desirable
to ask relevant questions about efficient transmission of energy.
We know that quantum mechanics offers collective effects on multiple
copies of a quantum state, as in ref. \cite{masson_universality_2022}
for example, and we would like to use them in designing
more efficient batteries.

Quantum batteries were first proposed by Alicki and Fannes \cite{alicki_entanglement_2013}; the original idea was to consider
a generic quantum Hamiltonian and observe how much work can be
extracted from it. An important figure of merit is the ergotropy
\cite{allahverdyan_maximal_2004}, which consists of the difference between the initial average
energy and the average energy after a unitary is implemented; this
unitary must minimize the difference.
Another relevant figure of merit, which we will focus on, is the
average charging power of the battery \cite{binder_quantacell_2015},
which is given by the energy as
\begin{equation}
P(t) = \frac{E(t)-E(0)}{t},
\end{equation}
where $t$ is time.

There have been several approaches to investigating this kind of
batteries. From Dicke models \cite{ferraro_high-power_2018}, spin
chains \cite{rossini_many-body_2019}, and based on the SYK model \cite{rossini_quantum_2020}.
These approaches are issued in closed quantum systems; only the
Hamiltonian of a battery system and a charger and an interaction is
needed.

There are, however, other approaches which are issued in the context
of open quantum systems, where an environment is considered, and the
interaction with it drives the main system \cite{breuer2002theory}. In this case, we consider
not only a Hamiltonian that operates in a system plus an interaction
but also dissipative operators which imply a loss of
energy to the environment. There can also be an input of energy
from the environment.
There are examples of
quantum batteries in this setting, from topological batteries \cite{Topological_Qua_Lu_Zh_2025} to
non-reciprocal ones \cite{Nonreciprocal_Q_Ahmadi_2024}. There are advantages of working with these
systems, as they constitute the most general way to describe a
quantum system consisting of a main system and environment.

Here, we use the open quantum system setting and define an
evolution of an operator in the interaction or Heisenberg picture.
We engineer specific evolutions, in this case Completely Positive
Unital (CPU) maps \cite{Heinosaari_Ziman_2011} which have a specific constant operator as its
fixed point \cite{2209.06806v1}. This means that after many implementations of the CPU
map, the result will be the desired constant operator. In other
words, the operator will be $A(t)=A_c t$, with $A_c$ a constant
operator. This means linear growth. If this operator is the
energy, we observe that $\langle A_c\rangle$ will correspond to the
average charging power of a battery. This is a general description
of these systems. As a concrete example, we give a general bound
on the maximum charging power a quantum battery can yield, when the
transmission between the environment and the battery is perfect.
We find that collective strategies can give a much better performance
for these batteries.
\section{CPU maps}
\label{sec:cpu}

The continuous evolution of an open quantum system can be described
with a master equation~\cite{breuer2002theory,wiseman_milburn_2009}. This global evolution, which assumes
openness to an environment, is
a generalization of Schrödinger’s equation, normally called the Lindblad equation or the Markovian quantum master equation. This Lindblad
equation fulfills the physical requirement that the evolution of a quantum state has to be
completely positive and trace preserving; in other words, it has to be described with a quantum channel \cite{WatrousTheTheoryof2018}.
Notice that this point of view is centered on
the quantum state, which is the object that evolves in time. There exists a dual picture
of the physics involved. Such a picture assumes that the observables evolve in time,
which is known as Heisenberg’s \cite{Heinosaari_Ziman_2011}. A corresponding open quantum system can be obtained in this
picture, which is known as the quantum Langevin equation or adjoint quantum master equation
\cite{wiseman_milburn_2009,breuer2002theory}. Here, the main point to notice is
that the evolution is given by the dual of a quantum channel.
Notice that applying the channel $\Phi$ to a state $\rho$ would produce, with
respect the observable $A$ the following
\begin{equation}
\tr[\Phi[\rho]A] = \tr[\rho\Phi^\dagger[A]],
\label{eq:primaldualmaps}
\end{equation}
where $\Phi^\dagger$ denotes the dual transformation for $\Phi$. We notice that the
numerical value is the same; the change in time is transformed from the state
to the observable. We notice that the adjoint transformation $\Phi^\dagger$ has to be
unital~\cite{Heinosaari_Ziman_2011}, as
\begin{equation}
\tr[\Phi[\rho]\mathds{I}] = \tr[\rho\Phi^\dagger[\mathds{I}]] = 1.
\end{equation}
As $\Phi$ is trace preserving,
\begin{equation}
\Phi^\dagger[\mathds{I}]=\mathds{I}.
\end{equation}
Therefore $\Phi^\dagger$ is unital. Also, as $(\Phi\otimes\mathds{I})^\dagger=\Phi^\dagger\otimes\mathds{I}$ then
$\Phi^\dagger$ is completely positive~\cite{Heinosaari_Ziman_2011}.
\section{$\Phi^\dagger$ with a given fixed point}
\label{sec:PhiFxPoint}

Suppose that a given observable $A$ on a Hilbert space is desired; one would
like to build completely positive and unital transformations $\Phi^\dagger$ that
have the desired observable $A$ as a fixed point. We can characterize this in
terms of its Choi representation. Suppose the Choi representation of a
completely positive and unital transformation $\Phi^\dagger$ is given by $Z_A$ then,
\begin{align}
    \tr_{\mathcal{H}2}[Z_A(\mathds{I}\otimes\mathds{I})] &= \mathds{I},\nonumber\\
Z_A &\geq 0.
\label{eq:cpucond}
\end{align}
The first condition asks for the transformation to be unital and the second asks
to be completely positive \cite{WatrousTheTheoryof2018}. If we want $Z_A$ to have
the observable $A$ as a fixed point, then
\begin{equation}
\Phi^\dagger[A]=\tr{\mathcal{H}2}[Z_A(\mathds{I}\otimes A^\intercal)] = A.
\label{eq:fixedcond}
\end{equation}
Given $A$ we can obtain a $Z_A$ that fulfills the fixed point condition (\ref{eq:fixedcond})
as well as the completely positive and unital conditions (\ref{eq:cpucond}).
Given $A\in\mathcal{H}$, then for a vector $\ket{v}\in\mathcal{H}$
\begin{equation}
Z_A \equiv A\otimes\frac{\ketbra{v}{v}^\intercal}{\bra{v}A\ket{v}}+\frac{\left(\mathds{I}-\frac{A}{\bra{v}A\ket{v}}\right)}{\left(\frac{N}{\tr A}-\frac{1}{\bra{v}A\ket{v}}\right)}\otimes\left(\frac{\mathds{I}}{\tr A}-\frac{\ketbra{v}{v}^\intercal}{\bra{v}A\ket{v}}\right)
\label{eq:ZSol}
\end{equation}
We can easily check that $Z_A$ fulfills the conditions of unitality (\ref{eq:cpucond}) and fixed point (\ref{eq:fixedcond}).
The positivity condition is a bit more delicate and thus treated as a theorem.
\begin{theorem}
$Z_A\geq0$ if and only if
\begin{align}
    A &\geq \mathds{I}\left(\frac{\bra{v}A\ket{v}-\tr A}{N-1}\right),\nonumber\\
\mathds{I}\bra{v}A\ket{v} &\geq A.
\end{align}
Where $N=\text{dim}\mathcal{H}$.
\end{theorem}
\begin{proof}
$\Rightarrow$ We multiply $Z_A$ given by Eq. (\ref{eq:ZSol}) by $\mathds{I}\otimes\ketbra{v}{v}^\intercal$
and get the trace in the second Hilbert space ($\mathcal{H}2$). This yields the inequality
\begin{equation}
\frac{A}{\bra{v}A\ket{v}}+\left(\mathds{I}-\frac{A}{\bra{v}A\ket{v}}\right)\frac{\frac{1}{\tr A}-\frac{1}{\bra{v}A\ket{v}~}}{\frac{N}{\tr A}-\frac{1}{\bra{v}A\ket{v}~}}\geq0.
\label{eq:vproy}
\end{equation}
After some algebra, we reach
\begin{equation}
A \geq \mathds{I}\left(\frac{\bra{v}A\ket{v}-\tr A}{N-1}\right).
\end{equation}
Then, we multiply $Z_A$ given by Eq. (\ref{eq:ZSol}) by $\mathds{I}\otimes\ketbra{v^\perp}{v^\perp}$ where
$\ket{v^\perp}$ is a vector such that $\braket{v^\perp}{v}=0$. We also get the trace in
$\mathcal{H}2$ to get
\begin{equation}
\frac{\mathds{I}-\frac{A}{\bra{v}A\ket{v}~}}{\frac{N}{\tr A}-\frac{1}{\bra{v}A\ket{v}~}}\frac{1}{\tr A}\geq0,
\label{eq:perpproy}
\end{equation}
therefore it is immediate to see
\begin{equation}
\mathds{I}\bra{v}A\ket{v}\geq A.
\end{equation}
$\Leftarrow$ Taking equation (\ref{eq:vproy}) and multiplying by $\otimes\ketbra{v}{v}^\intercal$ yields
\begin{equation}
\frac{A}{\bra{v}A\ket{v}}\otimes\ketbra{v}{v}^\intercal+\left(\mathds{I}-\frac{A}{\bra{v}A\ket{v}}\right)\otimes\ketbra{v}{v}^\intercal\frac{\frac{1}{\tr A}-\frac{1}{\bra{v}A\ket{v}~}}{\frac{N}{\tr A}-\frac{1}{\bra{v}A\ket{v}~}}\geq0.
\label{eq:Av}
\end{equation}
Multiplying Eq. (\ref{eq:perpproy}) by $\otimes\mathds{I}{/v}$ where $\mathds{I}{/v}+\ketbra{v}{v}^\intercal=\mathds{I}$,
yields
\begin{equation}
\frac{\mathds{I}-\frac{A}{\bra{v}A\ket{v}~}}{\frac{N}{\tr A}-\frac{1}{\bra{v}A\ket{v}~}}\frac{\otimes\mathds{I}{/v}}{\tr A}\geq0.
\label{eq:Aperp}
\end{equation}
Summing Eqs. (\ref{eq:Av}) and (\ref{eq:Aperp}) we finally get
\begin{equation}
A\otimes\frac{\ketbra{v}{v}^\intercal}{\bra{v}A\ket{v}}+\frac{\left(\mathds{I}-\frac{A}{\bra{v}A\ket{v}}\right)}{\left(\frac{N}{\tr A}-\frac{1}{\bra{v}A\ket{v}}\right)}\otimes\left(\frac{\mathds{I}}{\tr A}-\frac{\ketbra{v}{v}^\intercal}{\bra{v}A\ket{v}}\right)\geq0.
\end{equation}
\end{proof}
It should be stressed that the convergence is not to the observable $A$ but to a closely
related observable. We then introduce the following proposition.
\begin{proposition}
Given an initial observable $B$ if $\tr[A]\neq0$ and $\bra{v}A\ket{v}\neq0$ then
$\Phi^\dagger[B]$ is a fixed point of $\Phi^\dagger$, in other words, for all $B$
\begin{equation}
\Phi^\dagger[\Phi^\dagger[B]] =\Phi^\dagger[B].
\end{equation}
\label{thm:Prop1}
\end{proposition}
\begin{proof}
We first write $\Phi^\dagger[B]$ explicitly
\begin{equation}
\Phi^\dagger[B] = A\frac{\bra{v}B\ket{v}}{\bra{v}A\ket{v}}+\left(\mathds{1}-\frac{A}{\bra{v}A\ket{v}}\right)\frac{\left(\frac{\tr B}{\tr A}-\frac{\bra{v}B\ket{v}}{\bra{v}A\ket{v}~}\right)}{\left(\frac{N}{\tr A}-\frac{1}{\bra{v}A\ket{v}}\right)}.
\end{equation}
Now it is simple to observe that
\begin{align}
    \Phi^\dagger[\Phi^\dagger[B]] &= A\frac{\bra{v}B\ket{v}}{\bra{v}A\ket{v}}+\frac{A}{\bra{v}A\ket{v}}\cancel{\left(\tr \ketbra{v}{v}^\intercal-\frac{\bra{v}A\ket{v}}{\bra{v}A\ket{v}}\right)}\frac{\left(\frac{\tr B}{\tr A}-\frac{\bra{v}B\ket{v}}{\bra{v}A\ket{v}~}\right)}{\left(\frac{N}{\tr A}-\frac{1}{\bra{v}A\ket{v}}\right)}\nonumber\\
&+\left(\mathds{1}-\frac{A}{\bra{v}A\ket{v}}\right)\frac{\bra{v}B\ket{v}}{\bra{v}A\ket{v}}\frac{\cancel{\left(\frac{\tr A}{\tr A}-\frac{\bra{v}A\ket{v}}{\bra{v}A\ket{v}~}\right)}}{\left(\frac{N}{\tr A}-\frac{1}{\bra{v}A\ket{v}}\right)}\nonumber\\
&+\left(\mathds{1}-\frac{A}{\bra{v}A\ket{v}}\right)\frac{\left(\frac{\tr B}{\tr A}-\frac{\bra{v}B\ket{v}}{\bra{v}A\ket{v}~}\right)}{\left(\frac{N}{\tr A}-\frac{1}{\bra{v}A\ket{v}}\right)^2}\left(\frac{N}{\tr A}-\frac{\tr A}{\tr A\bra{v}A\ket{v}}-\cancel{\frac{1}{\bra{v}A\ket{v}~}}+\cancel{\frac{\bra{v}A\ket{v}}{\bra{v}A\ket{v}^2}}\right)\nonumber\\
&= A\frac{\bra{v}B\ket{v}}{\bra{v}A\ket{v}}+\left(\mathds{1}-\frac{A}{\bra{v}A\ket{v}}\right)\frac{\left(\frac{\tr B}{\tr A}-\frac{\bra{v}B\ket{v}}{\bra{v}A\ket{v}~}\right)}{\left(\frac{N}{\tr A}-\frac{1}{\bra{v}A\ket{v}}\right)}.
\end{align}
\end{proof}
Therefore, all further applications of the CPU map will yield the same operator $\Phi^\dagger[B]$.
\section{Fixed point observables}
\label{sec:FxPntObs}

Observe that we showed how to build completely positive unital transformations that have a
specific fixed point. This means we built a family of linear transformations in terms of $Z_A$
that fulfill the fixed point condition (\ref{eq:fixedcond}). The adjoint master equation
is a time-evolution equation of the type
\begin{equation}
\frac{dA}{dt}=\Phi^\dagger[A].
\label{eq:timeEv}
\end{equation}
In other words, the time evolution of the observable is given by a completely positive and
unital map \cite{Heinosaari_Ziman_2011,breuer2002theory}. Consider a discretization of the
time evolution (\ref{eq:timeEv}).
\begin{equation}
\frac{\Delta A}{\Delta t} = \Phi^\dagger[A].
\end{equation}
To go from time $t_0$ to time $t_n$ one would have to apply the channel $n$ times
\begin{equation}
\frac{\Delta_nA}{\Delta_nt} = \Phi^{\dagger n}[A].
\end{equation}
Following proposition \ref{thm:Prop1} we obtain that reiteration of a map $\Phi^\dagger$
yields the initial application. In other words for any observable $A_0$, $n>1$,
\begin{equation}
\Phi^{\dagger n}[A_0]=\Phi^\dagger[A_0].
\end{equation}
Therefore, recovering the continuous limit, we arrive at
\begin{equation}
\frac{dA}{dt}=\Phi^\dagger[A_0].
\label{eq:XfixA}
\end{equation}
If neither of $A_0$ or $\Phi^\dagger$ depend explicitly on $t$ we can integrate generically
Eq. (\ref{eq:XfixA}) and obtain
\begin{equation}
A = \Phi^\dagger[A_0]t.
\end{equation}
Therefore, there are observables that grow linearly in time. Observe now that the growth
$\langle A\rangle = \langle \Phi^\dagger[A_0]\rangle t$.
\section{Kraus operators}
\label{sec:KrausOps}
Remember that a quantum channel (CPTP map) can be written in terms of Kraus operators
\cite{WatrousTheTheoryof2018} as
\begin{equation}
\Phi[\rho]=\sum_iB_i\rho B_i^\dagger,
\end{equation}
for suitable operators $B_i$ which fulfill
\begin{equation}
\sum_iB_i^\dagger B_i=\mathds{I}.
\label{eq:completeness}
\end{equation}
Condition (\ref{eq:completeness}) is necessary for trace preservation.
Then, from equation (\ref{eq:primaldualmaps}) we observe
that
\begin{equation}
\tr[\Phi[\rho]A]=\tr[\sum_iB_i\rho B_i^\dagger A] =\tr[\rho \sum_iB_i^\dagger AB_i] = \tr[\rho\Phi^\dagger[A]],
\end{equation}
therefore,
\begin{equation}
\Phi^\dagger[A]=\sum_iB_i^\dagger AB_i.
\end{equation}
Now, remember that we are considering maps which are unital; therefore
we do not ask the fulfillment of condition (\ref{eq:completeness}) but an
analog condition
\begin{equation}
\Phi^\dagger[\mathds{I}]=\sum_iB_i^\dagger B_i=\mathds{I}.
\end{equation}
We should stress that there is no completeness condition in this
case as $\Phi^\dagger$ is not necessarily a CPTP map.
Observe that our CPU map is written in terms of
a Choi matrix in equation (\ref{eq:ZSol}). Using proposition 2.20 of \cite{WatrousTheTheoryof2018}
we can relate the Choi representation of linear maps to its Kraus
representation. For this purpose, we write a diagonalization of the observable $A$
as
\begin{equation}
A = \sum_ia_i\ketbra{a_i}{a_i}.
\end{equation}
We also complete the identity with a suitable orthonormal basis with respect to
$\ketbra{v}{v}$ as
\begin{equation}
\mathds{I}=\sum_j\ketbra{v_j}{v_j}.
\end{equation}
Let us define $\ketbra{v_r}{v_r}=\ketbra{v}{v}^\intercal$. Using these decompositions
and following the correspondence of Choi and Kraus representations, we
define the following operators
\begin{align}
    B_i^\dagger&= \sqrt{\frac{a_i}{\bra{v}A\ket{v}}~}\ketbra{a_i}{v}, \nonumber\\
C_{ij}^\dagger&= \sqrt{\left(\frac{1-\frac{a_i}{\bra{v}A\ket{v}}~}{\frac{N}{\tr A}-\frac{1}{\bra{v}A\ket{v}}~}\right)\left(\frac{1}{\tr A}-\frac{\delta_{jr}}{\bra{v}A\ket{v}}\right)}\ketbra{a_i}{v_j}.
\label{eq:KrausDuals}
\end{align}
We can check that the completeness relation (\ref{eq:completeness}) is fulfilled
\begin{equation}
\sum_iB_i^\dagger\mathds{I} B_i = \sum_i\frac{a_i}{\bra{v}A\ket{v}}\ketbra{a_i}{a_i}.
\end{equation}
Analogously we have
\begin{align}
    \sum_{ij} C_{ij}^\dagger\mathds{I} C_{ij} &=\sum_{ij}\left(\frac{1-\frac{a_i}{\bra{v}A\ket{v}}~}{\frac{N}{\tr A}-\frac{1}{\bra{v}A\ket{v}}~}\right)\left(\frac{1}{\tr A}-\frac{\delta_{jr}}{\bra{v}A\ket{v}}\right)\ketbra{a_i}{a_i} \nonumber\\
&= \sum_i\left(\frac{1-\frac{a_i}{\bra{v}A\ket{v}}~}{\frac{N}{\tr A}-\frac{1}{\bra{v}A\ket{v}}~}\right)\ketbra{a_i}{a_i}\left(\frac{N}{\tr A}-\frac{1}{\bra{v}A\ket{v}}\right)\nonumber\\
&= \mathds{I}-\sum_i\frac{a_i}{\bra{v}A\ket{v}}\ketbra{a_i}{a_i}.
\end{align}
Therefore, we obtain the unitality condition
\begin{equation}
\Phi^\dagger[\mathds{I}]=\sum_iB_i^\dagger B_i+\sum_{ij} C_{ij}^\dagger C_{ij} = \mathds{I}.
\end{equation}
\section{Example: upper bound on charging power}
\label{sec:TimeDil}

Let us consider a specific model of a quantum system that interacts with an
environment. As mentioned before, the most natural system for us to focus on is a
quantum battery modeled by a harmonic oscillator. The environment is modelled by
another harmonic oscillator, and we will have an interaction term, following \cite{Nonreciprocal_Q_Ahmadi_2024},
\begin{align}
H &= \omega_aa^\dagger a+\omega_bb^\dagger b +\varepsilon(e^{i\omega_Lt}a+e^{-i\omega_Lt}a^\dagger)
+\chi (a^\dagger b+b^\dagger a) \nonumber\\
&= H_0+H_I.
\end{align}
Where we define
\begin{equation}
H_I\equiv\chi(a^\dagger b+b^\dagger a).
\label{eq:Hinteraction}
\end{equation}
We define a state transformation
\begin{equation}
\tilde{\rho} \equiv e^{iH_0t}\rho(t)e^{-iH_0t}.
\end{equation}
A Von Neumann equation, which describes the general evolution of a quantum system
is given in the interaction picture as
\begin{equation}
\frac{d\tilde{\rho}}{dt} = -\frac{i}{\hbar}[\tilde{H},\tilde{\rho}].
\label{eq:vonNeum}
\end{equation}
Notice however that $H_I$ defined in (\ref{eq:Hinteraction}) only considers a transmission
of particles between the main system and the environment, which occur in a one-by-one manner.
This means that if the system loses a particle, the environment wins one and vice versa.
One would like to have a more general interaction Hamiltonian, which
would consider more general transfers of modes between the battery
(main system) and the environment. Something like
\begin{equation}
H_I^G \equiv \sum_{n,r}\chi_{n,r}({a^\dagger}^nb^n+{b^\dagger}^ra^r)
\label{eq:HIG}
\end{equation}
We thus would consider a higher number of transfers, two-by-two, three-by-three, etc., also mixes of
transfers (two-by-three, four-by-five etc).

We would ideally want a perfect transmission between the
environment and the battery.
Such unitary $U_p$ would operate on a state where the environment
is in level $n$ and the battery on level $0$ as
\begin{equation}
U_p\ket{0,n} = \ket{n,0},
\end{equation}
thus we end up the battery at level $N$.
Such a unitary would emulate the action of the interaction
Hamiltonian from Eq. (\ref{eq:HIG}). There is a physical process
that fulfills these requirements; we present it as an ansatz,
\begin{equation}
U = \sum_{n,r}\ketbra{n,r}{r,n}.
\end{equation}
Let us assume that this is the unitary that evolves the system in
the interaction picture for a unit of time. This is a unitary that
acts on the main (battery) system and the environment. Choosing a
suitable state for the environment, we can build a channel that
charges the battery constantly.
Suppose that we start with an environment in an initial state
\begin{equation}
\sigma = \sum_j\sigma_j\ketbra{j}{j}.
\end{equation}
Then applying the whole dynamics to the main system in state $\rho$ and then tracing the
environment part yields
\begin{equation}
\sum_i\bra{i}U\rho\otimes\sigma U^\dagger\ket{i} =\sum_{i,j}\sqrt{\sigma_j}\bra{i}U\ket{j}\rho\bra{j}U^\dagger\ket{i}\sqrt{\sigma_{j}}.
\end{equation}
Then, the Kraus operators of this evolution are given by
\begin{equation}
E_{i,j}^\dagger =\sqrt{\sigma_j}\sum_{n,r}\braket{j}{n}\braket{r}{i}\ketbra{r}{n}.
\end{equation}
Suppose that the starting observable is the energy of the battery,
i.e.
\begin{equation}
E_0 = \omega_b a^\dagger a.
\end{equation}
Observe then that applying the CPU map $\Phi^\dagger$ to $A_0$ yields
\begin{align}
    \Phi^\dagger &[E_0] = \nonumber\\
&\sum_{i,j}\sigma_j\left(\sum_{n,r}\braket{j}{n}\braket{r}{i}\ketbra{r}{n}\right)\omega_b a^\dagger a\left(\sum_{m,l}\braket{i}{l}\braket{m}{j}\ketbra{m}{l}\right).
\label{eq:Phiapplication}
\end{align}
Note that
\begin{equation}
\ketbra{r}{n}a^\dagger a\ketbra{m}{l} = \sqrt{nm}\ketbra{r}{l}\delta_{n-1,m-1}.
\end{equation}
Therefore, for $n=m=0$, the whole term goes to zero. We can thus do the
change $\delta_{n-1,m-1}\rightarrow\delta_{n,m}$. We get then
\begin{align}
\Phi^\dagger[E_0] &= \omega_b\sum_{j,n}\sigma_j n|\braket{j}{n}|^2\sum_{r,l}\delta_{r,l}\ketbra{r}{l},\nonumber \\
&= \omega_b\sum_{j,n}\sigma_j n|\braket{j}{n}|^2\mathds{1},\nonumber\\
&= \omega_b\varphi\mathds{1},
\label{eq:cqeye}
\end{align}
where $\varphi$ is a constant which depends on the state $\sigma$.
As $\Phi^\dagger$ sends the identity to the identity and is linear,
for any $k>0$,
\begin{equation}
\Phi^{\dagger k}[E_0] = \omega_b\varphi\mathds{1}.
\end{equation}

Remember that we are in the interaction picture. Nevertheless, the
transformation of the observable into the Schrödinger picture is
trivial, as the operator is proportional to the unitary and leaves
it the same. Therefore, the integration of the equation
\begin{equation}
\frac{d E}{dt} = \Phi^\dagger[E],
\label{eq:ObsTevol}
\end{equation}
yields for any time $t>0$ and any state $\rho$
\begin{equation}
\langle E(t)\rangle = \omega_b\varphi t.
\end{equation}

Each implementation of $\Phi^\dagger$ from (\ref{eq:Phiapplication})
implies the usage of a state $\sigma$. Therefore, the time evolution
from equation (\ref{eq:ObsTevol}) implies a constant flow of external
states.

Suppose that the basis $\ket{j}$ and $\ket{n}$ are the same. Then
we would get
\begin{equation}
\varphi = \sum_{n}\sigma_n n.
\end{equation}
Now, observe that if the initial state of the environment consists
in $\ket{1}^{\otimes N}$, which means $N$ qubits in the excited state
then the eigenvalues $\sigma_n$ reduce to zero in almost all places
except for one at the highest level. Therefore, in this case,
\begin{equation}
\varphi = 2^{N}-1,
\end{equation}
as the ground state is zero. We thus get an average charging power
$\langle P(t)\rangle = \omega_b(2^N-1)$, which is superextensive in
the number of systems considered.
\section{Discussion}
\label{sec:disc}
We have shown that using an open quantum system approach, a CPU map
which has a desired fixed point can be engineered. This map is
useful because the asymptotic behavior of the CPU map will be
known. The charging power in this case
is trivially calculated.

We give a specific example by using a unitary which interacts with
a system and an environment, interchanging the excitations among them.
We do not give the exact Hamiltonian that presents this behavior
but the unitary. We observe that more general interactions, which
allow the transmission of collective transitions, can yield a faster
charging of the quantum battery.

It is helpful to note that there is
much room for improvement in usual approaches of quantum batteries.
Being more explicit, usual approaches for quantum batteries are
implementations of one interchange of state between a charger and
a battery per time.  We observe that allowing multiple interchanges of 

states, collective effects allow a faster charging of the battery.
\bibliography{bibliography.bib}
\end{document}